\documentclass[sigconf]{acmart}
\usepackage{bm}
\usepackage{amsmath}
\usepackage{amsthm}
\usepackage{multirow}
\usepackage{xcolor}
\AtBeginDocument{%
 }

\copyrightyear{2026}
\acmYear{2026}
\setcopyright{cc}
\setcctype{by}
\acmConference[CIKM '26]{Proceedings of the 35th ACM International Conference on Information and Knowledge Management}{November 07--11, 2026}{Rome, Italy}
\acmBooktitle{Proceedings of the 35th ACM International Conference on Information and Knowledge Management (CIKM '26), November 07--11, 2026, Rome, Italy}
\acmDOI{10.1145/3799682.3840732}
\acmISBN{979-8-4007-2539-5/2026/11}

\begin{document}

\title{SWIM: Step-Wise Integrated Measure for Session-supervised List Evaluation in Generative Re-ranking}
\author{Yuanhao Pu}
\authornote{These authors contributed equally to this research.}
\authornote{Also affiliated with State Key Laboratory of Cognitive Intelligence}
\email{puyuanhao@mail.ustc.edu.cn}
\affiliation{
    \institution{University of Science and Technology of China}
    \city{Hefei}
    \state{Anhui}
    \country{China}
}

\author{Chenghao Zhang}
\authornotemark[1]
\email{zhangchenghao03@kuaishou.com}
\affiliation{%
  \institution{Kuaishou Technology}
  \city{Beijing}
  \country{China}
}

\author{Chao Feng}
\authornotemark[1]
\email{fengchao08@kuaishou.com}
\affiliation{%
  \institution{Kuaishou Technology}
  \city{Beijing}
  \country{China}}

\author{Xunyong Yang}
\email{yangxunyong@kuaishou.com}
\affiliation{%
  \institution{Kuaishou Technology}
  \city{Beijing}
  \country{China}
}

\author{Xiang Li}
\email{lixiang44@kuaishou.com}
\affiliation{%
 \institution{Kuaishou Technology}
 \city{Beijing}
 \country{China}
}

\author{Yongqi Liu}
\email{liuyongqi@kuaishou.com}
\affiliation{%
 \institution{Kuaishou Technology}
 \city{Beijing}
 \country{China}
}

\author{Defu Lian}
\authornotemark[2]
\email{liandefu@ustc.edu.cn}
\affiliation{
    \institution{University of Science and Technology of China}
    \city{Hefei}
    \state{Anhui}
    \country{China}
}

\author{Kaiqiao Zhan}
\email{zhankaiqiao@kuaishou.com}
\affiliation{%
 \institution{Kuaishou Technology}
 \city{Beijing}
 \country{China}
}

\author{Kun Gai}
\email{gai.kun@qq.com}
\affiliation{%
  \institution{Unaffiliated}
  \city{Beijing}
  \country{China}
}

\renewcommand{\shortauthors}{Pu et al.}

\begin{abstract}

Modern industrial recommender systems have increasingly adopted the Generator-Evaluator (G-E) framework for the re-ranking stage. Within this paradigm, the generator produces candidate item lists from a pool filtered by upstream retrieval and ranking modules, while the evaluator scores these lists and selects the highest-scoring one for final exposure per request. However, on sequential platforms (e.g., short-video apps), users consume items continuously, ignoring artificial list boundaries. Conventional evaluators score lists by aggregating point-wise values, implicitly assuming exposure independence. This fails to capture critical session-level dynamics, such as contextual dependencies, user continuation, and diminishing marginal utility from repetitive content.

To bridge this gap, we propose \textbf{SWIM} (\textbf{S}tep-\textbf{W}ise \textbf{I}ntegrated \textbf{M}easure), a list-level evaluator that models user behaviors as a finite-horizon prefix session-level survival process. SWIM estimates the prefix-conditioned contribution of the current list to the session-level objective by factorizing it into a recursive survival distribution and reached-position conditional rewards. Leveraging a causally-masked Transformer, SWIM efficiently estimates continuation probabilities and utilities in parallel, satisfying strict industrial latency constraints. Extensive experiments demonstrate that SWIM significantly outperforms baselines in listwise reranking tasks, yielding substantial improvements in overall recommendation engagement.
\end{abstract}




\maketitle

\section{Introduction}
\label{sec:intro}
\begin{figure*}[t]
    \centering
    \includegraphics[width=\textwidth]{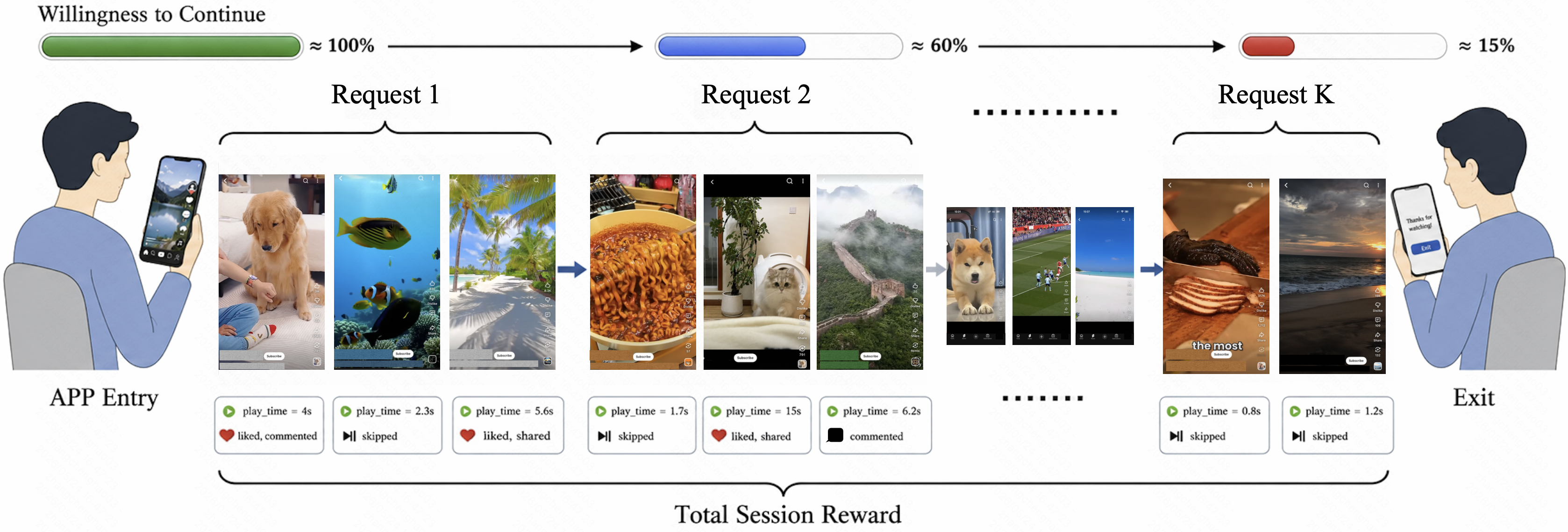}
    \caption{A recommendation session consists of multiple requests, the user's willingness to continue evolves with the preceding browsing history and accumulated feedback.}
    \label{fig:userbehav}
\end{figure*}

Recommender systems are widely deployed to alleviate information overload in diverse scenarios, such as short-video feeds~\cite{jeunen2023probabilistic}, news applications~\cite{moreira2018news,moreira2019contextual}, and e-commerce platforms~\cite{hidasi2016session,li2017neural}. To ensure both efficiency and effectiveness, modern industrial RS typically employs a multi-stage cascade architecture consisting of retrieval, ranking, and re-ranking. In this pipeline, the retrieval stage filters a massive item pool down to thousands of candidates, which are then scored by complex ranking models to yield a refined candidate set of hundreds of items. Finally, the re-ranking stage determines the ultimate sequence of items presented to the user~\cite{liu2022neural}. Recently, to better capture list-wise interactions, re-ranking has evolved into the Generator-Evaluator (G-E)~\cite{zhang2025from,feng2026flashevaluator,shi2023pier,yang2025comprehensive} framework. Under this generative paradigm, the generator produces a diverse collection of candidate lists from the refined item set, while the evaluator assesses each list and selects the optimal one for final exposure.

This paper focuses on the re-ranking evaluator within sequential content consumption platforms (e.g., short-video applications like TikTok). In such settings, items within a list are exposed sequentially, and the list for the next request is seamlessly presented once the current one is consumed. We define a session as the complete sequence of requests spanning from the initiation to the termination of a user's browsing activity. Although the underlying system inherently chunks items into discrete lists per request, users perceive an uninterrupted, boundary-less browsing feed. However, existing literature on re-ranking evaluators largely ignores this continuity. Conventional methods~\cite{zhang2025from,feng2026flashevaluator} evaluate candidate lists independently at the request level by simply accumulating point-wise utilities. By omitting crucial session-level contextual dependencies, these approaches fail to accurately model the user's continuation probability, leading to suboptimal list evaluation. Notably, our task fundamentally differs from traditional session-based recommendation (SBR)~\cite{wang2022survey,jannach2022session,quadrana2018sequence} and sequential next-item prediction (NTP) tasks~\cite{zhou2018deep,zhou2019deep,chen2019behavior}. While SBR and NTP aim to predict the interaction probability of the next individual item, our re-ranking evaluator is designed to assess the long-term cumulative value of a given, deterministic candidate list within a continuous session.

To address these limitations, we propose the \textbf{S}tep-\textbf{W}ise \textbf{I}ntegrated \textbf{M}easure (\textbf{SWIM}), a principled evaluation framework that explicitly captures the sequential continuity of user browsing by modeling it as a finite-horizon survival process. Specifically, SWIM reformulates the total list value as the expected cumulative utility over step-wise interaction prefixes. It mathematically decomposes this expected value into two core components: a \textbf{recursive survival distribution}, which models the probability of the user reaching each subsequent interaction step, and a \textbf{conditional utility function}, which estimates the reward generated at each reached position.

To ensure serving efficiency, we instantiate SWIM using a causally masked Transformer that estimates all step-wise continuation probabilities and utilities in a single forward pass for a given candidate list. This architectural design bypasses the need for explicit autoregressive rollouts across positions, enabling highly efficient parallel estimation. By circumventing the severe inference latency typically associated with autoregressive decoding, SWIM is ideally suited for real-world industrial deployment. Furthermore, extensive experiments on large-scale industrial datasets demonstrate that SWIM consistently outperforms strong baselines, driving significant improvements in long-horizon reward metrics such as total dwell time and overall session engagement.

\begin{figure*}[t]
    \centering
    \includegraphics[width=\textwidth]{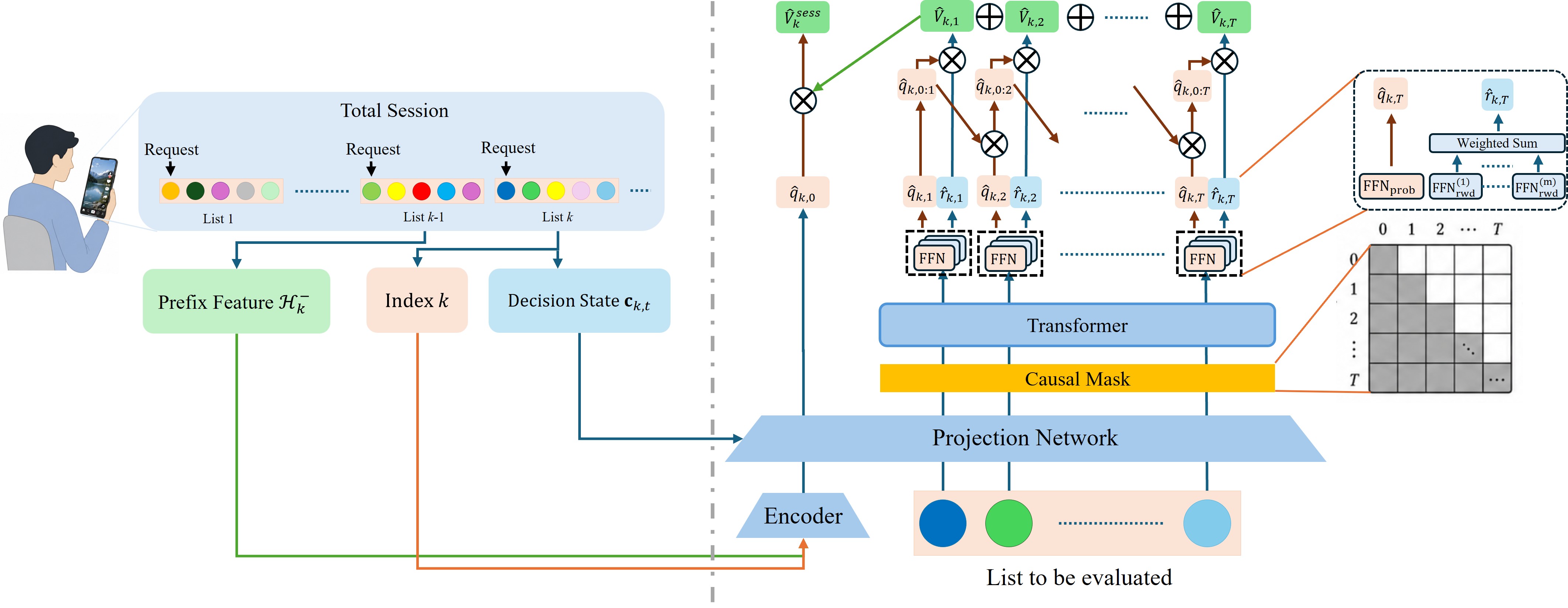}
    \caption{The model architecture of SWIM.}
    \label{fig:SWIM}
\end{figure*}

The main contributions of this work are summarized as follows:
\begin{itemize}
\item \textbf{Session-level value formulation under serving constraints.} We formalize recommendation evaluation from the perspective of prefix session-aware value and characterize the mismatch between session-level supervision in offline training and request-level inference in online serving.
\item \textbf{A survival-based structured decomposition.} We propose a finite-horizon survival-process formulation that decomposes prefix-session value into a recursive survival distribution and a cumulative reward function, enabling modeling of user continuation and long-horizon utility accumulation.
\item \textbf{A deployable non-autoregressive implementation.} We utilize a causal-masked Transformer architecture that supports parallel estimation of step-wise continuation probabilities and rewards, making the framework compatible with industrial latency requirements.
\item \textbf{Extensive industrial validation.} We validate SWIM on public datasets and a large-scale industrial platform, which shows significant improvements over strong baselines on offline metrics and online evaluation.
\end{itemize}
\section{Related Works}
In this section, we first review list-wise re-ranking approaches, with a particular focus on the Generator-Evaluator (G-E) framework in recommender systems. Next, we discuss session-based and sequential recommendation paradigms, highlighting the fundamental differences between these next-item prediction tasks and our list evaluation setting. Finally, we explore the application of survival and continuation analysis in recommender systems, which forms the theoretical foundation of our proposed method.

\subsection{Listwise Re-ranking and G-E Framework}

List-wise re-ranking~\cite{liu2022neural} aims to capture cross-item dependencies within a displayed list, overcoming the limitations of point-wise scoring. Existing methods broadly fall into generator-only (G-only) and G-E frameworks. G-only approaches~\cite{ai2018learning, bello2018seq2slate, pei2019personalized, pang2020setrank, meng2025generative, zhang2025goalrank} directly model list-wise contexts to output a single sequence but often struggle with the exploration of the vast permutation space.

To enhance exploration, modern systems adopt the G-E paradigm, where a generator proposes multiple candidate lists and an evaluator selects the optimal one. While recent works have significantly improved the generator's diversity~\cite{feng2021revisit, shi2023pier, yang2025comprehensive}, the evaluator design remains a bottleneck. Existing evaluators~\cite{zhang2025from, feng2026flashevaluator} typically score lists by aggregating request-level utilities, implicitly assuming independence across consecutive requests. Consequently, they fail to capture crucial session-level dynamics, such as cross-request dependencies and long-term user continuation.

The most relevant baseline CAVE~\cite{zhang2025from} also models user early-exit behavior. However, CAVE focuses strictly on the intra-request consumption value of the currently generated list. In contrast, SWIM targets the inter-request, session-level continuation value induced by the current list within the broader session. Technically, CAVE estimates current-list value by weighting sub-list values with user-specific exit probabilities, while SWIM rigorously parameterizes continuation via a recursive step-wise survival distribution. 

\subsection{Session-based and Sequential recommendation}
Session-based recommendation (SBR) aims to infer short-term user intent from current interaction sequences~\cite{wang2022survey,jannach2022session,quadrana2018sequence}. Early neural approaches like GRU4Rec~\cite{hidasi2016session} utilize recurrent networks, while subsequent works enhance session representations via attention mechanisms (e.g., NARM~\cite{li2017neural}, STAMP~\cite{liu2018stamp}) or graph structures (e.g., SR-GNN~\cite{wu2019session}). Beyond short sessions, sequential recommendation models incorporate long-term histories to capture evolving preferences, ranging from self-attentive architectures such as SASRec~\cite{kang2018self} and BERT4Rec~\cite{sun2019bert4rec} to scalable long-sequence models~\cite{pi2020search, chang2023twin, si2024twinv2, chai2025longer,zhang2026solar}. Target-aware sequence modeling has also been extensively validated in upstream ranking stages through models like DIN~\cite{zhou2018deep}, DIEN~\cite{zhou2019deep}, and BST~\cite{chen2019behavior}.

Despite their success in modeling sequential dependencies, they fundamentally operate from a next-item prediction (NTP) or item-level scoring perspective. They aim to predict the interaction probability of a single future item or estimate the point-wise relevance based on historical contexts. In contrast, SWIM tackles a fundamentally different list evaluation problem within a continuous session. Instead of asking "\textit{what item to recommend next}", SWIM assesses "\textit{how a deterministically given candidate list contributes to the long-term cumulative value of the entire session}". Consequently, directly applying SBR or NTP models fails to capture the complex intra-list interactions and the session-level continuation dynamics that our proposed framework explicitly models.

\subsection{Survival and Continuation Modeling in Recommendation}
Survival analysis provides a principled framework for modeling time-to-event and continuation processes, and has been increasingly explored in recommender systems. For instance, \cite{jing2017neural} apply survival modeling to predict user return times and subsequent activities. In the context of sequential feeds, C-3PO~\cite{jeunen2023probabilistic} models a user's scrolling budget via a survival function to derive personalized, position-aware exposure probabilities.
While these approaches share the core intuition of characterizing user continuation and early-exit behaviors, they primarily operate at the item or position level. SWIM advances this line from position-level estimation to list-wise session value evaluation. SWIM embeds survival-style behavioral modeling into the G-E re-ranking paradigm, achieving session-supervised optimization without violating strict request-level latency constraints.
\section{Preliminaries}
\label{sec:preliminaries}

Consider a recommendation session as a continuous user interaction process. Let \(\{v_t\}_{t\ge 1}\) denote the exposure-consumption stream within a session, where \(v_t\) is the item at the \(t\)-th position of the session trajectory. 
Let \(\tau\in\mathbb{N}_{+}\) be a discrete stopping time representing the first position that is not reached or consumed by the user. 
Equivalently, the user interacts with exactly the prefix \((v_1,\dots,v_{\tau-1})\), and the session terminates right before \(\tau\). 

For each position \(t\), let \(R_t\in\mathbb{R}_{+}\) denote the reward generated at that position, which may summarize one or multiple engagement signals such as click, like, share, follow, and play time. 
We assume that the total expected utility is finite and well-defined.

\begin{definition}[Total Session Value]
The total value of a recommendation session is defined as the cumulative reward collected before the session terminates, i.e., $G=\sum_{t\ge 1}\mathbf{1}\{\tau > t\} R_t$. Therefore, we consider that the session-level objective is to optimize the expected total session value $V^{\mathrm{sess}}=\mathbb{E}[G]$.
\end{definition}

This formulation describes the user-centric behavioral process over a whole session. 
In industrial recommender systems, however, the serving interface is \textbf{request-based}: due to limited latency constraint, the inference process cannot generate or optimize a list that is long enough to serve the entire session trajectory. Instead, the system serves the user through several request-level decision epochs. At each epoch \(k\), the system receives a candidate set from upstream retrieval or ranking stages and selects an ordered list with fixed length for the current request. 
Formally, let $\mathcal{V}_k = \{v_i\}_{i=1}^{N_k}$ denote the candidate set transmitted from upstream stage at request \(k\). 
The re-ranking model selects a length-\(T\) ordered list 
\begin{equation}
    L_k = (v_{k,1},v_{k,2},\dots,v_{k,T}) \in \Pi_T(\mathcal{V}_k),
\end{equation}
where \(\Pi_T(\mathcal{V}_k)\) denotes the set of all length-\(T\) ordered lists constructed from \(\mathcal{V}_k\). Let \(\mathcal{H}^{-}_k\) denote the observable session prefix before request \(k\), including historical interactions, previous request feedback, cumulative in-session engagement, request index, and contextual information. 
Let \(\mathbf{c}_k\in\mathcal{C}\) denote the decision state at \(k\), which summarizes the information of the current request, including candidate-side features, user features, and request-level context.

Under this formulation, \(L_k\) is a local list of the global session trajectory. The re-ranking problem is therefore not only to estimate the isolated utility of a standalone list, but also to characterize how the current list affects user survival and contributes to session-level objective through its prefix-conditioned current-list value.

\subsection{Generator-Evaluator Re-ranking}

Many practical re-ranking systems can be described through a G-E interface~\cite{feng2021revisit,shi2023pier,yang2025comprehensive}. Given the candidate set \(\mathcal{V}_k\) and decision state \(\mathbf{c}_k\), a generator proposes a set of feasible candidate lists:
\begin{equation}
\mathcal{L}_k
=
\mathcal{G}_{\theta}(\mathcal{V}_k,\mathbf{c}_k)
=
\{L_k^{(1)},L_k^{(2)},\dots,L_k^{(M)}\},
\end{equation}
where each \(L_k^{(m)}\in\Pi_T(\mathcal{V}_k)\). 
An evaluator then assigns a listwise score to each generated list $s_k^{(m)}=\mathcal{E}_{\phi}(L_k^{(m)},\mathbf{c}_k)$. The final exposed list is selected by
\begin{equation}
L_k^{*}
=
\arg\max_{L_k^{(m)}\in\mathcal{L}_k}
\mathcal{E}_{\phi}(L_k^{(m)},\mathbf{c}_k).
\end{equation}

In this work, SWIM focuses on the evaluator-side problem: estimating a session-aware cumulative value for a candidate list under request-level serving constraints.

\subsection{Survival Process and Reward}

We next introduce the notations to describe user survival process and expected reward within a request-level list. Let $A_k\in\{0,1\}$ denote whether the user reaches request \(k\). 
The request-entry probability is defined as 
\begin{equation}
    q_{k,0}=\mathbb{P}(A_k=1\mid \mathcal{H}^{-}_k).
\end{equation}
This quantity describes the probability that the user enters the current request from the preceding session prefix. 
For the first request in a session, we set \(q_{1,0}=1\) by convention.

Conditioned on \(A_k=1\), let $Y_{k,t}\in\{0,1\}, t=1,\dots,T+1$ denote whether the user reaches position \(t\) within the current list. 
Here, \(Y_{k,t}=1\) for \(t\le T\) means that the user reaches the \(t\)-th item in \(L_k\), while \(Y_{k,T+1}=1\) indicates that the user finishes the current list and proceeds beyond the current list. 
Conditioned on the current request being reached, the first position is reached by convention:
\begin{equation}
Y_{k,1}=1 \quad \text{given }\quad A_k=1.
\end{equation}
Besides, the browsing process satisfies the monotonicity constraint
\begin{equation}
Y_{k,t}=0 \Rightarrow Y_{k,t+1}=0,
\quad t=1,\dots,T.
\end{equation}

For each position \(t=1,\dots,T\), the in-list step-wise \textbf{survival probability} is defined as
\begin{equation}
q_{k,t}(L_k\mid \mathbf{c}_k)
=
\mathbb{P}
\left(
Y_{k,t+1}=1
\mid
Y_{k,t}=1,\mathbf{c}_k,L_{k,1:t}
\right),
\end{equation}
where \(L_{k,1:t}=(v_{k,1},\dots,v_{k,t})\). 
This quantity represents the probability that the user continues beyond position \(t\), conditioned on having reached that position and observed the prefix \(L_{k,1:t}\). The corresponding conditional in-list survival probability is defined as
\begin{equation}
S_{k,t}(L_k\mid \mathbf{c}_k)
=
\mathbb{P}
\left(
Y_{k,t}=1
\mid
A_k=1,\mathbf{c}_k,L_{k,1:t-1}
\right).
\end{equation}
Note that $S_{k,1} = 1$ and
\begin{equation}
S_{k,t}(L_k\mid \mathbf{c}_k)
=
\prod_{i=1}^{t-1}
q_{k,i}(L_k\mid \mathbf{c}_k),
\quad t=2,\dots,T+1.
\end{equation}

Combining request-entry survival and in-list survival, the session-wise probability of reaching position \(t\) in the current list is
\begin{equation}
S^{\mathrm{sess}}_{k,t}(L_k\mid \mathcal{H}^{-}_k,\mathbf{c}_k)
=
q_{k,0}
S_{k,t}(L_k\mid \mathbf{c}_k)=
\prod_{i=0}^{t-1}
q_{k,i}(L_k\mid \mathbf{c}_k).
\end{equation}

\begin{definition}[Conditional Reward]
Let \(R_{k,t}\in\mathbb{R}_{+}\) denote the reward generated by the item \(v_{k,t}\) at position \(t\) of the current list. Define the conditional reward as
\begin{equation}
r_{k,t}(L_k\mid \mathbf{c}_k)
=
\mathbb{E}
\left[
R_{k,t}
\mid
Y_{k,t}=1,\mathbf{c}_k,L_{k,1:t}
\right].
\end{equation}
In practice, \(R_{k,t}\) may correspond to a single feedback signal or a weighted combination of multiple engagement signals.
\end{definition}

Based on the above survival and reward variables, we define the value of the current request-level list. 
\begin{definition}[Conditional Current-list Value]
The conditional value of list \(L_k\) at request \(k\), conditioned on the current request being reached, is defined as
\begin{equation}
V_k(L_k\mid \mathbf{c}_k)
=
\sum_{t=1}^{T}
S_{k,t}(L_k\mid \mathbf{c}_k)
r_{k,t}(L_k\mid \mathbf{c}_k).
\end{equation}
This value measures the expected value generated within the current list, where each position-wise reward is weighted by the probability that the user reaches that position.
\end{definition}

\begin{definition}[Session-wise Current-list Contribution]
The session-wise contribution of list \(L_k\) is defined as
\begin{equation}\label{eq:V_sess}
\begin{aligned}
V^{\mathrm{sess}}_k(L_k\mid \mathcal{H}^{-}_k,\mathbf{c}_k)
&=
q_{k,0}
V_k(L_k\mid \mathbf{c}_k)\\&=
\sum_{t=1}^{T}
S^{\mathrm{sess}}_{k,t}(L_k\mid \mathcal{H}^{-}_k,\mathbf{c}_k)
r_{k,t}(L_k\mid \mathbf{c}_k).
\end{aligned}
\end{equation}
\end{definition}

The session-wise value \(V^{\mathrm{sess}}_k\) additionally accounts for the probability that request \(k\) is reached from the preceding session prefix.
\section{Methodology}
\label{sec:methodology}

In this section, we describe how SWIM utilizes these quantities to be parameterized, trained, and used for online re-ranking. As shown in Fig.~\ref{fig:SWIM}, the SWIM evaluator estimates the session-wise value contributed by the current request-level chunk in Eq.~\ref{eq:V_sess}. At request \(k\), the generator proposes a set of candidate lists
\[
\mathcal{L}_k=\{L_k^{(1)},L_k^{(2)},\dots,L_k^{(M)}\}.
\]
SWIM acts as the evaluator and assigns a score to each candidate list. During training, SWIM learns the session-wise score through probability and conditional reward supervision,
\begin{equation}
\hat V^{\mathrm{sess}}_k(L_k^{(m)}\mid \mathcal{H}^{-}_k,\mathbf{c}_k)
=
\hat q_{k,0}
\hat V_k(L_k^{(m)}\mid \mathbf{c}_k).
\end{equation}
For online inference, however, the re-ranking model is invoked after request \(k\) has already been sent. Therefore, the local decision is set to \(q_{k,0}=1\) and the final exposed list is therefore simplified by
\begin{equation}
L_k^{*}
=
\arg\max_{L_k^{(m)}\in\mathcal{L}_k}
\hat V_k(L_k^{(m)}\mid \mathbf{c}_k).
\end{equation}

\subsection{Prefix Encoder}

For each candidate list proposed by the generator, SWIM encodes the observable prefix features and the candidate list into hidden states. Let $\mathbf{g}_k=\mathrm{Encoder}(\mathcal{H}^{-}_k)$ denote a prefix representation constructed from user history, request index, previous request feedback, cumulative in-session engagement, and contextual features.

To unify request-entry modeling and in-list survival modeling, we introduce a prefix token at position \(0\). 
This token represents the session prefix before the current list is entered:
\begin{equation}
x_{k,0}
=
\mathrm{Proj}
\left(
\mathbf{g}_k
\right),
\end{equation}
where $\mathrm{Proj}$ is a projection network that maps features into hidden dimension \(d\). For each item embedding \(v_{k,t}\) in the candidate list, we construct the item representation as
\begin{equation}
x_{k,t}
=
\mathrm{Proj}
\left(
[
v_{k,t};
\mathbf{g}_k;
\mathbf{e}^{\mathrm{pos}}_t;
\mathbf{c}_{k,t}
]
\right),
\quad t=1,\dots,T.
\end{equation}
where \(\mathbf{e}^{\mathrm{pos}}_t\) is the position embedding. Stacking the prefix token and all item representations gives
\begin{equation}
x_k
=
[x_{k,0},x_{k,1},\dots,x_{k,T}]^\top
\in\mathbb{R}^{(T+1)\times d}.
\end{equation}

Since \(q_{k,t}\) and \(r_{k,t}\) are defined with respect to the prefix \(L_{k,1:t}\), SWIM utilizes a \textbf{causal-masked Transformer} to encode prefix-aware list representations. 
The prefix token at position \(0\) is used to estimate the boundary survival probability \(q_{k,0}\), while the item token at position \(t\) is used to estimate \(q_{k,t}\) and \(r_{k,t}\). To ensure prefix-causal encoding, we use a lower-triangular causal attention mask
\begin{equation}
\mathbf{M}_{i,j}
=
\begin{cases}
0, & j\le i,\\
-\infty, & j>i,
\end{cases}
\quad i,j\in\{0,\dots,T\}.
\end{equation}
The masked attention is then computed as
\begin{equation}
\mathrm{Attn}(\mathbf{Q},\mathbf{K},\mathbf{V})
=
\mathrm{softmax}
\left(
\frac{\mathbf{Q}\mathbf{K}^{\top}}{\sqrt{d}}
+
\mathbf{M}
\right)
\mathbf{V}.
\end{equation}

After the Transformer layers, we obtain the hidden states
\begin{equation}
\mathbf{H}_k
=
[\mathbf{h}_{k,0},\mathbf{h}_{k,1},\dots,\mathbf{h}_{k,T}]^\top.
\end{equation}
where \(\mathbf{h}_{k,0}\) summarizes session prefixs before current request, and \(\mathbf{h}_{k,t}\) summarizes decision states and the list prefix \(L_{k,1:t}\) for \(t\ge 1\).

\subsection{Survival Estimation}
For each position \(t=0,\dots,T\), SWIM estimates the step-wise survival probability by applying a shared sigmoid probability head:
\begin{equation}
\hat q_{k,t}
=
\sigma
\left(
\mathrm{FFN}_{\mathrm{prob}}(\mathbf{h}_{k,t})
\right),
\quad t=0,\dots,T.
\end{equation}
Here, \(\hat q_{k,0}\) is predicted from the prefix hidden state \(\mathbf{h}_{k,0}\), and \(\hat q_{k,t}\) for \(t\ge 1\) is predicted from the list-prefix hidden state \(\mathbf{h}_{k,t}\). The session-wise survival probability of reaching position \(t\) in the current list is recovered by cumulative products:
\begin{equation}
\hat S^{\mathrm{sess}}_{k,t}
=
\prod_{i=0}^{t-1}
\hat q_{k,i},
\quad t=1,\dots,T+1.
\end{equation}
For each logged request, let \(y_{k,t}\in\{0,1\}\) denote whether the user reaches position \(t\). We define \(y_{k,0}=1\) by convention, since the session prefix exists. The boundary label \(y_{k,1}\) indicates whether request \(k\) is reached from the preceding prefix. We set \(y_{k,1}=1\) for observed requests and \(y_{k,1}=0\) for a sampled terminal request, padded to length \(T\) if necessary. For \(t\ge 1\), \(y_{k,t+1}\) indicates whether the user continues beyond position \(t\). The loss is defined as
\begin{equation}
\mathcal{L}_{\mathrm{prob}}
=
\sum_{k}
\sum_{t=0}^{T}
y_{k,t}
\,
\operatorname{BCE}
\left(
y_{k,t+1},
\hat q_{k,t}
\right),
\end{equation}
which ensures that the survival loss at position \(t\) is valid only when the user has reached that position. 
For \(t=T\), the label \(y_{k,T+1}\) indicates whether the user finishes the current list and proceeds beyond the current chunk.

\subsection{Conditional Reward Estimation}

For each reached position, SWIM estimates the conditional reward generated by the corresponding item. 
Reward prediction is defined only for item positions \(t=1,\dots,T\), not for the prefix token at \(t=0\).

For binary engagement objectives such as click, like, follow, or share, we use target-specific sigmoid heads:
\begin{equation}
\hat r^{(m)}_{k,t}
=
\sigma
\left(
\mathrm{FFN}^{(m)}_{\mathrm{rwd}}(\mathbf{h}_{k,t})
\right),
\quad t=1,\dots,T,
\end{equation}
where \(m\) indexes a specific engagement target. For continuous objectives such as play time, direct regression is usually unstable due to long-tailed distributions. We therefore discretize the target into \(B\) ordered buckets and predict a categorical distribution:
\begin{equation}
\hat{\boldsymbol{\pi}}_{k,t}
=
\mathrm{softmax}
\left(
\mathrm{FFN}_{\mathrm{pt}}(\mathbf{h}_{k,t})
\right)
\in\mathbb{R}^{B}.
\end{equation}
Let \(\{w_b\}_{b=1}^{B}\) denote bucket representatives. 
The conditional expected play-time is reconstructed as 
\begin{equation}
\hat r^{\mathrm{pt}}_{k,t}=\sum_{b=1}^{B}w_b\hat{\pi}_{k,t,b}.
\end{equation}
When continuous and binary feedback are both used, \(\hat r_{k,t}\) is obtained by combining all selected reward heads with task-specific weights.

Reward heads are trained only on reached positions, since \(r_{k,t}\) is defined as a conditional reward given that the user reaches position \(t\). 
For a binary feedback target \(m\), let \(z^{(m)}_{k,t}\in\{0,1\}\) denote the observed feedback label. 
The binary reward loss is
\begin{equation}
\mathcal{L}^{(m)}_{\mathrm{rwd}}
=
\sum_{k}
\sum_{t=1}^{T}
y_{k,t}
\,
\mathrm{BCE}
\left(
z^{(m)}_{k,t},
\hat r^{(m)}_{k,t}
\right).
\end{equation}
For bucketized play-time prediction, let \(b_{k,t}\in\{1,\dots,B\}\) be the bucket label. 
The corresponding loss is
\begin{equation}
\mathcal{L}_{\mathrm{pt}}
=
\sum_{k}
\sum_{t=1}^{T}
y_{k,t}
\,
\mathrm{CE}
\left(
b_{k,t},
\hat{\boldsymbol{\pi}}_{k,t}
\right).
\end{equation}

We combine all reward losses as
\begin{equation}
\mathcal{L}_{\mathrm{rwd}}
=
\sum_{m=1}^{M}
\alpha_m
\mathcal{L}^{(m)}_{\mathrm{rwd}}
+
\alpha_{\mathrm{pt}}
\mathcal{L}_{\mathrm{pt}}.
\end{equation}

\subsection{List Scoring and Training Objective}

Given the estimated survival probabilities and conditional rewards, SWIM computes the session-wise current-list contribution as
\begin{equation}
\hat V^{\mathrm{sess}}_k(L_k\mid \mathcal{H}^{-}_k,\mathbf{c}_k)
=
\sum_{t=1}^{T}
\hat S^{\mathrm{sess}}_{k,t}
\hat r_{k,t}
=
\sum_{t=1}^{T}
\left(
\prod_{i=0}^{t-1}
\hat q_{k,i}
\right)
\hat r_{k,t}.
\end{equation}
This score includes the boundary transition \(\hat q_{k,0}\), and is used when session-wise request-entry calibration is available.

For online re-ranking, the evaluator is invoked only after request \(k\) has been triggered. Therefore, the boundary transition is fixed to \(\hat q_{k,0}=1\), and candidate lists are compared by the conditional value
\begin{equation}
\hat V_k(L_k\mid \mathbf{c}_k)
=
\sum_{t=1}^{T}
\hat S_{k,t}
\hat r_{k,t}
=
\sum_{t=1}^{T}
\left(
\prod_{i=1}^{t-1}
\hat q_{k,i}
\right)
\hat r_{k,t}.
\end{equation}
The final exposed list is selected as
\begin{equation}
L_k^{*}
=
\arg\max_{L_k^{(m)}\in\mathcal{L}_k}
\hat V_k(L_k^{(m)}\mid \mathbf{c}_k).
\end{equation}

The overall training objective contains only the survival probability loss and the conditional reward loss:
\begin{equation}
\mathcal{L}
=
\mathcal{L}_{\mathrm{prob}}
+
\mathcal{L}_{\mathrm{rwd}}.
\end{equation}
This objective directly supervises the two main components of SWIM: the boundary-augmented survival chain \(\hat q_{k,0:T}\) and the reached-position conditional rewards \(\hat r_{k,1:T}\). 

\section{Theoretical Results}
\label{sec:theory}

In this section, we theoretically analyze why the boundary-augmented survival parameterization of SWIM is suitable for session-supervised re-ranking. First, SWIM has a session-prefix representation advantage: by introducing the boundary transition \(q_0\), it can represent prefix-conditioned session-wise list contribution, while request-level evaluators like CAVE~\cite{zhang2025from} model the consumption value after the current request is entered.
Then, unlike CAVE's exit component, which is parameterized by a Weibull distribution, SWIM uses a distribution-free discrete survival parameterization that can represent any valid finite-horizon survival behavior.

For clarity, we omit the request index \(k\) in this section. We introduce \(Y_1=1\) indicates that the current request is reached from the preceding prefix. For \(t=1,\dots,T\), the browsing process is monotone, i.e., $Y_t=0 \Rightarrow Y_{t+1}=0$.
SWIM parameterizes the boundary-augmented continuation chain as
\begin{equation}
\begin{aligned}
q_0
&=
\mathbb{P}(Y_1=1\mid \mathcal{H}^{-}), \\
q_t
&=
\mathbb{P}
\left(
Y_{t+1}=1
\mid
Y_t=1,\mathcal{H}^{-},L_{1:t}
\right),\ t=1,\dots,T.
\end{aligned}
\end{equation}
SWIM estimates the session-wise current-list contribution as
\begin{equation}
V_{\mathrm{SWIM}}^{\mathrm{sess}}(L,\mathcal{H}^{-})
=
\sum_{t=1}^{T}
\left(
\prod_{i=0}^{t-1}q_i
\right)r_t.
\end{equation}
When the current request has been reached, the boundary transition is fixed to \(q_0=1\), and the request-level current-list value becomes
\begin{equation}
V_{\mathrm{SWIM}}^{\mathrm{req}}(L)
=
\sum_{t=1}^{T}
\left(
\prod_{i=1}^{t-1}q_i
\right)r_t.
\end{equation}

\subsection{Session-prefix Advantage}

We first show that the boundary transition \(q_0\) gives SWIM a representation advantage for session-wise list contribution. 
This advantage is independent of the particular neural architecture used to estimate \(q_0\). 
It follows from the fact that session-wise contribution depends not only on the value of the current list after the request is reached, but also on whether the request is reached from the preceding session prefix.

Let \(\mathbf{c}\) denote request-level information available to a request-only evaluator, such as the current candidate list and request-level context, and \(\mathcal{H}^{-}\) denote additional session-prefix information. Assume the true session-wise current-list contribution has the form
\begin{equation}
V_{\mathrm{sess}}^{*}
=
q_0^{*}(\mathcal{H}^{-},\mathbf{c})
V_{\mathrm{req}}^{*}(\mathbf{c}),
\end{equation}
where \(q_0^{*}(\mathcal{H}^{-},\mathbf{c})\in[0,1]\) is the true probability that the request is reached from the preceding prefix, and \(V_{\mathrm{req}}^{*}(\mathbf{c})\ge 0\) is the true request-level value conditioned on the current request being reached.

\begin{theorem}[Irreducible error of request-only evaluators]
Consider any request-only evaluator \(g(\mathbf{c})\) that does not depend on the preceding session prefix \(\mathcal{H}^{-}\). 
Under squared loss, the minimum achievable prediction risk for the session-wise target \(V_{\mathrm{sess}}^{*}\) is
\begin{equation}
\inf_g
\mathbb{E}
\left[
\left(
g(\mathbf{c})-V_{\mathrm{sess}}^{*}
\right)^2
\right]
=
\mathbb{E}
\left[
\mathrm{Var}
\left(
q_0^{*}(\mathcal{H}^{-},\mathbf{c})
V_{\mathrm{req}}^{*}(\mathbf{c})
\mid \mathbf{c}
\right)
\right].
\end{equation}
Equivalently,
\begin{equation}
\inf_g
\mathbb{E}
\left[
\left(
g(\mathbf{c})-V_{\mathrm{sess}}^{*}
\right)^2
\right]
=
\mathbb{E}
\left[
\left(
V_{\mathrm{req}}^{*}(\mathbf{c})
\right)^2
\mathrm{Var}
\left(
q_0^{*}(\mathcal{H}^{-},\mathbf{c})
\mid \mathbf{c}
\right)
\right].
\end{equation}
Therefore, if \(V_{\mathrm{req}}^{*}(\mathbf{c})>0\) and the boundary transition \(q_0^{*}\) varies across session prefixes with the same request-level information \(\mathbf{c}\), then every request-only evaluator has strictly positive irreducible error. SWIM can explicitly model the boundary transition \(q_0\) from \(\mathcal{H}^{-}\).
\end{theorem}

\begin{proof}
For any function \(g(\mathbf{c})\), the optimal predictor of \(V_{\mathrm{sess}}^{*}\) under squared loss is the conditional expectation
\begin{equation}
g^{*}(\mathbf{c})
=
\mathbb{E}
\left[
V_{\mathrm{sess}}^{*}\mid \mathbf{c}
\right].
\end{equation}
The corresponding minimum risk is the conditional variance:
\begin{equation}
\inf_g
\mathbb{E}
\left[
\left(
g(\mathbf{c})-V_{\mathrm{sess}}^{*}
\right)^2
\right]
=
\mathbb{E}
\left[
\mathrm{Var}
\left(
V_{\mathrm{sess}}^{*}\mid \mathbf{c}
\right)
\right].
\end{equation}
Substituting
\[
V_{\mathrm{sess}}^{*}
=
q_0^{*}(\mathcal{H}^{-},\mathbf{c})
V_{\mathrm{req}}^{*}(\mathbf{c})
\]
and noting that \(V_{\mathrm{req}}^{*}(\mathbf{c})\) is fixed when conditioning on \(\mathbf{c}\), we obtain
\begin{equation}
\mathrm{Var}
\left(
V_{\mathrm{sess}}^{*}\mid \mathbf{c}
\right)
=
\left(
V_{\mathrm{req}}^{*}(\mathbf{c})
\right)^2
\mathrm{Var}
\left(
q_0^{*}(\mathcal{H}^{-},\mathbf{c})
\mid \mathbf{c}
\right).
\end{equation}
Thus, whenever \(q_0^{*}\) has non-zero conditional variance given \(\mathbf{c}\), a request-only evaluator cannot eliminate the prediction error for the session-wise target. 
SWIM avoids this limitation by conditioning \(q_0\) on the preceding session prefix.
\end{proof}

This clarifies the difference between SWIM and request-level evaluators. SWIM fits session-wise signals that vary with previous browsing history, accumulated engagement, and user fatigue.

\subsection{Prior-free Survival Modeling}
\label{sec:theory_expressiveness}

We then compare the expressiveness of SWIM's discrete survival formulation with the stochastic exit prior used by CAVE. 
CAVE decomposes user exit probability into an interest-driven component and a stochastic component, where the stochastic component is modeled by a continuous-time Weibull distribution to capture external factors~\cite{zhang2025from}. 
This Weibull component provides a useful inductive bias when the true stochastic exit behavior is smooth and Weibull-like. 
However, it also restricts the stochastic survival pattern to a low-dimensional parametric family.

On a finite list horizon, the Weibull survival curve sampled at discrete positions can be written as
\begin{equation}
S_t^{W}(\lambda,\kappa)
=
\exp\left(
-\left(\frac{t-1}{\lambda}\right)^{\kappa}
\right),
\quad \lambda>0,\ \kappa>0,
\end{equation}
where \(\lambda\) is the scale parameter and \(\kappa\) is the shape parameter. 
The induced step-wise continuation probability is
\begin{equation}
q_t^{W}(\lambda,\kappa)
=
\frac{S_{t+1}^{W}(\lambda,\kappa)}
{S_t^{W}(\lambda,\kappa)}
=
\exp\left(
-\frac{t^{\kappa}-(t-1)^{\kappa}}{\lambda^{\kappa}}
\right).
\end{equation}
The Weibull exit prior defines a two-dimensional survival manifold
\begin{equation}
\mathcal{M}_{W}
=
\left\{
q^{W}_{1:T}(\lambda,\kappa):
\lambda>0,\kappa>0
\right\}
\subset [0,1]^{T}.
\end{equation}

In contrast, SWIM directly parameterizes the discrete continuation sequence
\begin{equation}
q_{0:T}=(q_0,q_1,\dots,q_T)\in[0,1]^{T+1},
\end{equation}
where \(q_0\) is the boundary transition from the session prefix to the current request and \(q_t\), \(t\ge 1\), is the in-list continuation probability. 
Therefore, SWIM is not restricted to any pre-specified time-to-exit family. 
It can represent arbitrary finite-horizon continuation patterns, including position-irregular drops, abrupt contextual disruptions, repeated content, or session-stage-specific fatigue.

\begin{theorem}[Expressiveness gap of Weibull-restricted survival]
Let \(q^*_{1:T}\in[0,1]^{T}\) be the data-generating continuation process, and let $S_t^*
=\prod_{i=1}^{t-1}q_i^*$ be the corresponding survival profile. Let \(q^W_{1:T}\in\mathcal{M}_W\) be the best Weibull-restricted projection of \(q^*_{1:T}\) under a chosen fitting criterion, and define the local approximation residual
\begin{equation}
\delta_j=q_j^* - q_j^W.
\end{equation}
Then for any $t$, the discrepancy between the true survival profile and the Weibull-restricted survival profile admits the exact expansion
\begin{equation}
S_t^* - S_t^W
=
\sum_{j=1}^{t-1}
\delta_j
S_j^*
\frac{S_t^W}{S_{j+1}^W}.
\label{eq:survival_product_difference}
\end{equation}
Consequently, for the current-list utility
\begin{equation}
V(q)
=
\sum_{t=1}^{T}S_t(q)r_t,
\end{equation}
the value estimation discrepancy satisfies
\begin{equation}
V(q^*)-V(q^W)
=
\sum_{j=1}^{T-1}
\delta_j
\frac{S_j^*}{S_{j+1}^W}
\left(
\sum_{t=j+1}^{T}
S_t^W r_t
\right).
\label{eq:value_bias_decomposition}
\end{equation}
\end{theorem}

\begin{proof}
We first prove the survival discrepancy identity. 
For a fixed \(t\), define
\begin{equation}
B_j
=
\left(
\prod_{i=1}^{j-1}q_i^*
\right)
\left(
\prod_{i=j}^{t-1}q_i^W
\right),
\quad j=1,\dots,t.
\end{equation}
Then $B_1=S_t^W, B_t=S_t^*$. By telescoping,
\begin{equation}
S_t^*-S_t^W=B_t-B_1=\sum_{j=1}^{t-1}(B_{j+1}-B_j).
\end{equation}
For each \(j\),
\begin{equation}
B_{j+1}-B_j
=
(q_j^*-q_j^W)
\left(
\prod_{i=1}^{j-1}q_i^*
\right)
\left(
\prod_{i=j+1}^{t-1}q_i^W
\right).
\end{equation}
Using \(\delta_j=q_j^*-q_j^W\), we have
\[
\prod_{i=1}^{j-1}q_i^*=S_j^*,\quad
\prod_{i=j+1}^{t-1}q_i^W
=
\frac{S_t^W}{S_{j+1}^W},
\]
we obtain
\begin{equation}
S_t^*-S_t^W
=
\sum_{j=1}^{t-1}
\delta_j
S_j^*
\frac{S_t^W}{S_{j+1}^W}.
\end{equation}
Then we substitute this identity into the value difference:
\begin{align}
V(q^*)-V(q^W)
=
\sum_{t=1}^{T}
(S_t^*-S_t^W)r_t =
\sum_{t=1}^{T}
\left(
\sum_{j=1}^{t-1}
\delta_j
S_j^*
\frac{S_t^W}{S_{j+1}^W}
\right)
r_t.
\end{align}
Exchanging the order of summation gives
\begin{equation}
V(q^*)-V(q^W)
=
\sum_{j=1}^{T-1}
\delta_j
\frac{S_j^*}{S_{j+1}^W}
\left(
\sum_{t=j+1}^{T}
S_t^W r_t
\right),
\end{equation}
which proves Eq.~\eqref{eq:value_bias_decomposition}. 
\end{proof}

The theorem reveals how local Weibull misspecification propagates into list evaluation. 
A local approximation residual \(\delta_j\) at step \(j\) is not isolated: it is weighted by the downstream residual value
$\sum_{t=j+1}^{T}S_t^W r_t$. Therefore, an early survival mismatch can affect the estimated contribution of all downstream positions. 
When the true continuation pattern contains non-Weibull deviations, such as abrupt position-level drop-offs or prefix-dependent irregularities, the low-dimensional Weibull manifold \(\mathcal{M}_W\) cannot generally set all residuals \(\delta_j\) to zero simultaneously.

This can be formalized as a representation-level approximation gap. The Weibull family \(\mathcal{M}_W\) is a two-dimensional smooth manifold embedded in the \((T+1)\)-dimensional discrete continuation space \([0,1]^{T+1}\). For \(T+1>2\), \(\mathcal{M}_W\) is a strict lower-dimensional subset which cannot cover all valid finite-horizon continuation sequences. Consequently, for a generic continuation process \(q^*_{0:T}\notin\mathcal{M}_W\), every Weibull-restricted estimator has non-zero structural residuals.

SWIM avoids this approximation bias by operating in the full finite-horizon discrete survival space. 
For any valid survival profile
\begin{equation}
1=S_0\ge S_1\ge\cdots\ge S_{T+1}\ge 0,
\end{equation}
there exists a discrete continuation sequence \(q_{0:T}\in[0,1]^{T+1}\) s.t.
\begin{equation}
S_t
=
\prod_{i=0}^{t-1}q_i,
\quad t=1,\dots,T+1.
\end{equation}
Specifically, whenever \(S_t>0\), one can set
\begin{equation}
q_t
=
\frac{S_{t+1}}{S_t};
\end{equation}
if \(S_t=0\), downstream \(q_t\) values can be chosen arbitrarily because all later survival probabilities are already zero. 
Thus, in the finite-horizon setting considered by request-level re-ranking, SWIM can represent any valid discrete survival profile.

This does not imply that the Weibull prior is always harmful. 
When the true stochastic exit behavior is indeed smooth and Weibull-like, the parametric prior may provide useful regularization and improve sample efficiency. 
However, SWIM has a strictly more flexible structural parameterization when user continuation is
driven by irregular effects. 
It can reduce the parametric misspecification risk introduced by a fixed Weibull stochastic exit prior, while remaining fully compatible with session-supervised list evaluation.
\section{Experiments}
\label{sec:experiments}

The experiments are designed to answer the following questions:

\begin{itemize}
    \item \textbf{RQ1:} Does the proposed evaluator (SWIM) improve re-ranking performance compared with representative pointwise, listwise, G-only, and G-E baselines?
    \item \textbf{RQ2:} How important are the key components of SWIM, including prefix conditioning,  survival modeling, and conditional reward modeling?
    \item \textbf{RQ3:} Can SWIM bring measurable online improvements in a real industrial system?
\end{itemize}
The code is available at \url{https://github.com/yuanhao53/SWIM}.

\subsection{Datasets}
We evaluate SWIM empirically on RecFlow and KuaiRand, and conduct online experiments on Kuaishou's real traffic. As public datasets lack complete prefix features, we select RecFlow and KuaiRand for their maximal coverage of required signals among available benchmarks (Table~\ref{tab:dataset_feature_coverage}).

\textbf{RecFlow}~\cite{liu2025recflow} is a public industrial recommendation dataset collected from a real-world recommendation pipeline. 
It provides request-level or stage-level recommendation signals, candidate items, user feedback, and multi-type engagement labels. 
Since RecFlow contains relatively rich request-level information, it is the closest public dataset to the generator-evaluator re-ranking setting considered in this work.

\textbf{KuaiRand}~\cite{gao2022kuairand} is a public short-video recommendation dataset. It contains user behavior sequences, timestamps, user/item features, random exposure signals, and multiple user feedback signals. 

\textbf{Kuaishou}'s online experiment is conducted on Kuaishou APP with more than 400 million DAUs. Different from public datasets, it contains real online request-level candidate sets, generated slates, item positions, previous-request feedback, session-prefix features, multi-objective rewards, and online serving logs. Therefore, it provides the most faithful evaluation environment for SWIM.

\begin{table}[t]
\centering
\caption{Availability of SWIM-required signals in public datasets. 
\(\checkmark\) indicates directly available, \(\triangle\) indicates partially available or reconstructed, and \(\times\) indicates unavailable.}
\label{tab:dataset_feature_coverage}
\resizebox{\linewidth}{!}{
\begin{tabular}{lcc}
\toprule
\textbf{SWIM-required signal} & \textbf{RecFlow} & \textbf{KuaiRand} \\
\midrule
User ID / user history & \(\checkmark\) & \(\checkmark\) \\
Request-level grouping & \(\checkmark\) & \(\triangle\) \\
Ordered list & \(\checkmark\) & \(\triangle\) \\
Candidate set for re-ranking & \(\checkmark\) & \(\triangle\) \\
Multi-stage serving information & \(\checkmark\) & \(\times\) \\
Previous-request feedback & \(\triangle\) & \(\triangle\) \\
Cumulative in-session engagement & \(\triangle\) & \(\triangle\) \\
In-list continuation signal \(q_{k,t}\) & \(\checkmark\) & \(\triangle\) \\
Binary feedback, e.g., click/like/share & \(\checkmark\) & \(\checkmark\) \\
Continuous feedback, e.g., play/view time & \(\checkmark\) & \(\checkmark\) \\
\bottomrule
\end{tabular}
}
\end{table}

\begin{table}[t]
\centering
\caption{Dataset statistics after task construction.}
\label{tab:dataset_statistics}
\resizebox{\linewidth}{!}{
\begin{tabular}{lcccccc}
\toprule
\textbf{Dataset} & \textbf{\#User} & \textbf{\#Session} & \textbf{\#Requests} & \textbf{\#Items} & \textbf{Cand. size} & \textbf{List length} \\
\midrule
RecFlow & 34,574 & 965,634 & 3,308,233 & 14,181,768 & 120 & 6 \\
KuaiRand & 27,285 & 3,339,426 & 54,805,516 & 31,922,917 & 6 & 6 \\ 
\bottomrule
\end{tabular}
}
\end{table}

\subsection{Task Construction}

At each request \(k\), a re-ranking model receives a candidate set \(\mathcal{V}_k\) and outputs an ordered list \(L_k=(v_{k,1},\dots,v_{k,T})\). In all datasets, we set the target list length to \(T=6\).

For \textbf{RecFlow}, we directly use the candidate set provided by the upstream ranking stage. 
Each request contains \(|\mathcal{V}_k|=120\) candidate items, and the model selects and orders a length-6 list from this candidate set. This setting is closest to the standard generator-evaluator re-ranking protocol.
For \textbf{KuaiRand}, the upstream candidate set is not available. 
We therefore segment each user's behavior sequence into sessions according to timestamp and further split each session into fixed-length request lists. Since only the observed interaction sequence is available, we construct a \(6\)-to-\(6\) re-ordering task: each list contains six observed items, and the model evaluates different orderings of these items. This setting is weaker than full candidate-set re-ranking, but it still allows us to evaluate whether SWIM can better model in-list continuation and reward under sequential short-video consumption.

For \textbf{Kuaishou}'s online experiments, we directly use the candidate sets, generated lists, item positions, and request-level context features from the production serving logs.

During training, we use all available feedback signals to provide supervision. For each reached item position, we construct binary targets for engagements such as video view, click, like, follow, comment, or share when available. For continuous feedback such as play time, we discretize the signal into several ordered buckets and use the bucket label to train the play-time prediction head. 

\subsection{Baselines}

We compare SWIM with representative baselines from pointwise ranking, listwise re-ranking, G-only re-ranking, and G-E re-ranking:

\begin{itemize}
    \item \textbf{DNN} independently predicts the score of each candidate item and outputs the final list obtained by sorting.
    \item \textbf{Seq2Slate}~\cite{bello2018seq2slate} formulates a seq-to-seq slate generation problem which generates the output list autoregressively.
    \item\textbf{DLCM}~\cite{ai2018learning} encodes the local ranking context of the input list and adjusts item scores accordingly. 
    \item\textbf{SetRank}~\cite{pang2020setrank} uses self-attention to model interactions among candidates and learns permutation-invariant listwise scores. 
    \item \textbf{PRM}~\cite{pei2019personalized} is a Transformer-based re-ranking model that captures item-item interactions within the candidate list.
    \item \textbf{SORT-Gen}~\cite{meng2025generative} 
    is an efficient generative re-ranking model which uses a Sequential Ordered Regression Transformer to estimate multi-objective values for variable-length sub-lists, and a mask-driven fast generation algorithm to produce the final list efficiently. We implement a SORT-Gen-style generative re-ranking baseline following the main design. 
    \item \textbf{NAR4Rec}~\cite{ren2024nonautoregressive} is a non-autoregressive generative re-ranking model that generates the target list in parallel rather than decoding items one by one. 
    \item \textbf{PIER}~\cite{shi2023pier} is a permutation-level interest-based end-to-end re-ranking framework. 
    It follows a two-stage architecture: a permutation selection module first generates candidate permutations according to permutation-level user interest, and a context-aware prediction module then evaluates the generated permutations. 
    \item \textbf{MultiG}~\cite{yang2025comprehensive} is a multi-generator re-ranking framework that uses multiple complementary generators to produce more comprehensive and diverse candidate lists before evaluation. 
    \item \textbf{CAVE}~\cite{zhang2025from} is a consumption-aware list value estimation Evaluator for generator-evaluator re-ranking. 
    It considers that users may exit before consuming the full generated list and estimates list value by weighting sub-list values with user-specific exit probabilities. Both CAVE and SWIM utilize the same Generator architecture as MultiG.

\end{itemize}

For fair comparison, all methods follow the same data construction and evaluation protocol. 
We split each dataset by request timestamp, using the earliest \(90\%\) of requests for training and the latest \(10\%\) for prediction/testing. 
This temporal split prevents future information leakage and mimics the online serving setting.  For G-E methods including MultiG, CAVE and SWIM, we utilize a combination of Generators including DNN, Seq2Slate, PRM and NAR4Rec, which produces a total of 20 candidate lists (5 of each) with length $T = 6$ from the same candidate pool (120 for RecFlow and 6 for KuaiRand), then the Evaluator scores all of them and output Top-1. 

\subsection{Evaluation Metrics}

We report two offline metrics computed w.r.t. the video-view label.

\textbf{NDCG}(@6) evaluates the quality of the final ordered list. 
For each request, we use the binary video-view label \(z^{\mathrm{view}}\) as the relevance signal and compute NDCG over the $n=6$ positions. 
This metric measures whether items that receive positive video-view feedback are ranked higher in the final list.

\textbf{AUC} evaluates the model's ability to distinguish video-view-positive and video-view-negative candidate items. 
For each request, candidate items are scored by the corresponding model, and AUC is computed using \(z^{\mathrm{view}}\) as the binary label. 
For generative or generator-evaluator methods that output an ordered list, we use the scalar ranking or evaluator score used to produce the final ordering.

\subsection{Main Results}

Table~\ref{tab:main_results} reports the offline re-ranking performance on RecFlow and KuaiRand. 
SWIM achieves the best performance on both datasets and both metrics, consistently outperforming pointwise, listwise, G-only, and G-E baselines. These results indicate that explicitly modeling user continuation and reached-position conditional reward provides a stronger evaluator for request-level re-ranking.

\begin{table*}[t]
\centering
\caption{Offline re-ranking performance on RecFlow and KuaiRand dataset. 
The best result is shown in bold.}
\label{tab:main_results}
\begin{tabular}{lcccc}
\toprule
\multirow{2}{*}{\textbf{Model}}
& \multicolumn{2}{c}{\textbf{RecFlow}}
& \multicolumn{2}{c}{\textbf{KuaiRand}} \\
\cmidrule(lr){2-3}
\cmidrule(lr){4-5}
& \textbf{NDCG} & \textbf{AUC}
& \textbf{NDCG} & \textbf{AUC}\\
\midrule
DNN       & $0.1584 \pm 0.0004$ & $0.6591 \pm 0.0003$ &	$0.6907 \pm 0.0003$ & $0.7150 \pm 0.0003$ \\
Seq2Slate & $0.1693 \pm 0.0023$ & $0.6732 \pm 0.0021$ & $0.6984 \pm 0.0018$ & $0.7184 \pm 0.0020$  \\
DLCM      & $0.1747 \pm 0.0006$ & $0.6771 \pm 0.0003$ & $0.7025 \pm 0.0002$ & $0.7252 \pm 0.0003$ \\
SetRank   & $0.1823 \pm 0.0006$ & $0.6845 \pm 0.0009$ & $0.7095 \pm 0.0012$ & $0.7481 \pm 0.0014$  \\
PRM       & $0.1841 \pm 0.0005$ & $0.6866 \pm 0.0008$ & $0.7132 \pm 0.0008$ & $0.7554 \pm 0.0007$  \\
SORT-Gen       & $0.1892 \pm 0.0007$ & $0.6973 \pm 0.0011$ & $0.7204 \pm 0.0010$ & $0.7632 \pm 0.0016$  \\
NAR4Rec   & $0.1793 \pm 0.0011$ & $0.6816 \pm 0.0017$ & $0.7076 \pm 0.0020$ & $0.7358 \pm 0.0016$  \\
PIER   & $0.1873 \pm 0.0014$ & $0.6907 \pm 0.0017$ & $0.7186 \pm 0.0025$ & $0.7562 \pm 0.0018$  \\
MultiG    & $0.1911 \pm 0.0018$ & $0.7024 \pm 0.0014$ & $0.7231 \pm 0.0019$ & $0.7677 \pm 0.0012$  \\
\midrule
CAVE      & $0.1957 \pm 0.0013$ & $0.7083 \pm 0.0016$ & $0.7313 \pm 0.0016$ & $0.7756 \pm 0.0010$  \\
\textbf{SWIM}       & $\mathbf{0.2031 \pm 0.0015}$ & $\mathbf{0.7185 \pm 0.0017}$ & $\mathbf{0.7375 \pm 0.0011}$ & $\mathbf{0.7804 \pm 0.0013}$
           \\
\bottomrule
\end{tabular}
\end{table*}

\textbf{Analysis.}
According to the offline results, listwise re-ranking methods such as DLCM, SetRank, and PRM consistently outperform the pointwise DNN baseline, confirming that modeling cross-item dependencies and positional context is important for listwise re-ranking. Besides, recent generative and generator-evaluator methods, including PIER, SORT-Gen, and MultiG, further improve by exploring richer permutation spaces or by using stronger list-level evaluators. CAVE achieves the strongest performance among all baselines, suggesting that explicitly considering user early-exit behavior is beneficial for list evaluation. SWIM further improves over CAVE, which meets our theoretical understandings of its advantage.

Besides, we also observe that the improvement on RecFlow is larger than that on KuaiRand. 
This is expected because RecFlow provides richer request-level and candidate-set information, making it closer to the generator-evaluator re-ranking scenario considered by SWIM. 
KuaiRand provides fewer request-level serving signals and requires reconstructed request chunks, but SWIM still obtains consistent gains, indicating that survival-weighted reward aggregation remains effective even under weaker data approximations.

\subsection{Ablation Study}

We conduct ablation studies on RecFlow to analyze the contribution of SWIM's main components. 
Specifically, we examine whether the performance gain comes from session-prefix modeling, boundary transition modeling, or the survival-based decomposition itself.

\begin{itemize}
    \item \textbf{w/o Prefix Features}: remove the session-prefix features from the prefix encoder. The model still receives the current candidate list and item-side features, but cannot explicitly condition on the preceding session context.

    \item \textbf{w/o Boundary \(q_{k,0}\)}: remove the boundary transition from the session prefix to the current request. The survival chain starts from the first item in the current list, and the model only estimates in-list continuation probabilities \(q_{k,1:T}\).

    \item \textbf{w/o Survival}: remove the recursive survival chain and directly aggregate the predicted conditional rewards. This variant keeps the reward heads but does not weight each position by its estimated reachability probability, reducing SWIM to a reward-aggregation evaluator.
\end{itemize}

\begin{table}[t]
\centering
\caption{Ablation study on RecFlow dataset.}
\label{tab:ablation}
\begin{tabular}{lcc}
\toprule
\textbf{Variant} & \textbf{NDCG} & \textbf{AUC} \\
\midrule
w/o Prefix Features      & 0.1975 & 0.7093 \\
w/o Boundary \(q_{k,0}\) & 0.2014 & 0.7128 \\
w/o Survival   & 0.1952 & 0.7071 \\
\midrule
\textbf{SWIM}                & \textbf{0.2031} & \textbf{0.7185} \\
\bottomrule
\end{tabular}
\end{table}

\textbf{Analysis.}
Table~\ref{tab:ablation} shows that each component contributes to the final performance. 
The ablations verify the necessity of SWIM's design choices: prefix conditioning improves session-aware representation, the boundary transition introduces prefix-to-request survival supervision, and the recursive survival chain provides the key reachability-aware weighting mechanism for list evaluation.

\subsection{Online A/B Test}
We conducted a 7-day online A/B test on Kuaishou's main feed reranking pipline (>400M DAU, >2 hour avg. stay time per user each day) with 5\% traffic allocation. The control group follows the production re-ranking pipeline, while the treatment replaces CAVE~\cite{zhang2025from} with SWIM, leaving upstream retrieval, ranking, and generator modules intact. Specifically, the generator produces $\sim$70 lists from 60 candidates, and the evaluator selects the optimal one. Relative improvements over the baseline are reported in Table~\ref{tab:online_ab}.

\begin{table}[t]
\centering
\caption{Online A/B test results on the Industrial reranking system. 
Relative lifts are reported against the baseline.}
\label{tab:online_ab}
\resizebox{\linewidth}{!}{
\begin{tabular}{lcc}
\toprule
\textbf{Metric} & \textbf{Relative lift} & \textbf{Confidence Interval} \\
\midrule
App stay time & +0.351\% &
[0.27\%, 0.42\%] \\
% Video play time & +0.732\% &
% [0.64\%, 0.82\%] \\
Retention (LT7) &  +0.048\% & [0.01\%, 0.09\%] \\
\bottomrule
\end{tabular}
}
\end{table}

\textbf{Analysis.}The online A/B test evaluates whether offline ranking improvements translate into tangible gains in real-world user engagement and retention. We report two core online metrics: App stay time and the 7-day retention rate (LT7). App stay time measures the total duration a user spends within the application. LT7 denotes the 7-day user lifetime, which serves as a key indicator for evaluating long-term DAU stability and retention benefits in the subsequent week. Both metrics exhibit positive relative lifts, with 95\% confidence intervals reported.
\section{Conclusion}

In this paper, we proposed SWIM, a session-supervised evaluator for generative re-ranking. 
SWIM models a candidate list through a boundary-augmented survival process, which captures both the transition from the session prefix to the current request and the step-wise continuation behavior within the list. By combining recursive survival probabilities with reached-position conditional rewards, SWIM estimates the prefix-conditioned contribution of the current list to the session-level objective.

Looking forward, the SWIM framework opens promising avenues for future research. One critical direction involves developing robust counterfactual and off-policy evaluation protocols for generative list spaces. Furthermore, extending our session-prefix modeling to accommodate longer and more heterogeneous browsing trajectories presents an exciting challenge.

\appendix
\section*{Acknowledgements}
This work was supported by Kuaishou Technology and the National Natural Science Foundation of China (No. U24A20253).
\section*{GenAI Usage Disclosure}
Generative AI was used solely for editorial purposes, including grammar checking, sentence refinement, and improving conciseness. It was not used for generating content or making substantive intellectual contributions.

\bibliographystyle{ACM-Reference-Format}
\bibliography{bib}

\end{document}